\documentclass[twosided,final]{article}
\usepackage[utf8]{inputenc}
\usepackage{graphicx}
\usepackage{fullpage}
\usepackage[english]{babel} 
\usepackage{amssymb,mathtools}
\usepackage{xspace}
\usepackage{booktabs}
\usepackage{boxedminipage}
\usepackage{tabularx}
\usepackage{calc}
\usepackage{hyperref}
\usepackage{csquotes}
\usepackage{color}
\usepackage{microtype}
\usepackage{enumitem}
\usepackage{bm}

\usepackage{tikz}
\usetikzlibrary{arrows,positioning,backgrounds,fit,shapes,calc,scopes}

\usepackage[textwidth=20mm,obeyFinal,colorinlistoftodos]{todonotes}

\usepackage{amsthm}
\theoremstyle{plain}
\newtheorem{theorem}{Theorem}
\newtheorem{corollary}[theorem]{Corollary}

\newtheorem{lemma}[theorem]{Lemma}

\theoremstyle{definition}
\newtheorem{definition}[theorem]{Definition}
\newtheorem*{example*}{Example}

\DeclareMathOperator{\poly}{{\rm poly}}

\DeclareMathOperator{\rank}{{\rm rank}}

\newcommand{\FPT}{{\sf FPT}\xspace}
\newcommand{\NP}{{\sf NP}\xspace}

\newcommand{\transpose}{\intercal}
\newcommand{\hy}{\hbox{-}\nobreak\hskip0pt}

\def\ve#1{\mathchoice{\mbox{\boldmath$\displaystyle\bf#1$}}
{\mbox{\boldmath$\textstyle\bf#1$}}
{\mbox{\boldmath$\scriptstyle\bf#1$}}
{\mbox{\boldmath$\scriptscriptstyle\bf#1$}}}
\newcommand\vea{{\ve a}}

\newcommand\veb{{\ve b}}
\newcommand\vecc{{\ve c}}

\newcommand\vece{{\ve e}}

\newcommand\ven{{\ve n}}
\newcommand\vem{{\ve m}}

\newcommand\ves{{\ve s}}

\newcommand\vex{{\ve x}}
\newcommand\vey{{\ve y}}

\newcommand\veDelta{{\ve \Delta}}

\def\Z{\mathbb{Z}}
\def\N{\mathbb{N}}

\def \P {\mathcal{P}}

\def \RR {\mathcal{R}}

\def \l {\langle}
\def \r {\rangle}
\newcommand{\pref}{\ensuremath{\succ}}
\newcommand{\Oh}{O}

\newcommand{\CC}{\mathcal{C}}

\newcommand{\prob}[3]{
    \noindent
    \begin{center}
      \fbox{
        \begin{minipage}{.96\linewidth}
          \begin{tabularx}{\columnwidth}{lX}
	        \multicolumn{2}{l}{#1}\\
	        {\bf Input:}&{#2}\\
	        {\bf Find:}&{#3}
          \end{tabularx}
        \end{minipage}
      }
    \end{center}
}

\newcommand{\fpt}{fixed\hy{}pa\-ra\-me\-ter trac\-ta\-ble\xspace}

\newcommand{\OhOp}[1]{O\mathopen{}\mathclose\bgroup\left( #1 \aftergroup\egroup\right)}
\newcommand{\heading}[1]{\smallskip\textit{#1.}}

\begin{document}
\title{%
A Unifying Framework for Manipulation Problems
\thanks{Research supported by CE-ITI grant project P202/12/G061 of GA~{\v C}R, project NFR MULTIVAL, and ERC Starting Grant 306465 (BeyondWorstCase).
}}

\date{}

\author{
Dušan Knop\thanks{\texttt{dusan.knop@uib.no} Department of Applied Mathematics, Charles University, Prague, Czech Republic \emph{and} Department of Informatics, University of Bergen, Bergen, Norway}
\and
Martin Koutecký\thanks{\texttt{koutecky@technion.ac.il} Technion - Israel Institute of Technology, Haifa, Israel}
\and
Matthias Mnich\thanks{\texttt{mmnich@uni-bonn.de} Institut f{\"u}r Informatik, Universit\"at Bonn, Bonn, Germany \emph{and} Department of Quantitative Economics, Maastricht University, Maastricht, The~Netherlands}
}

\bibliographystyle{alpha}      
\maketitle
\begin{abstract}
  Manipulation models for electoral systems are a core research theme in social choice theory; they include bribery (unweighted, weighted, swap, shift, \dots), control (by adding or deleting voters or candidates), lobbying in referenda and others.

  We develop a unifying framework for manipulation models with few types of people, one of the most commonly studied scenarios.
  A critical insight of our framework is to separate the descriptive complexity of the voting rule $\mathcal R$ from the number of types of people.
  This allows us to finally settle the computational complexity of $\mathcal R$-{\sc Swap Bribery}, one of the most fundamental manipulation problems.
  In particular, we prove that $\mathcal R$-\textsc{Swap Bribery} is fixed-parameter tractable when $\mathcal R$ is Dodgson's rule and Young's rule, when parameterized by the number of candidates.
  This way, we resolve a long-standing open question from 2007 which was explicitly asked by Faliszewski et al.~[JAIR 40, 2011].

  Our algorithms reveal that the true hardness of bribery problems often stems from the complexity of the voting rules.
  On one hand, we give a fixed-parameter algorithm parameterized by number of types of people for complex voting rules.
  Thus, we reveal that $\mathcal R$-\textsc{Swap Bribery} with Dodgson's rule is much harder than with Condorcet's rule, which can be expressed by a conjunction of linear inequalities, while Dodson's rule requires quantifier alternation and a bounded number of disjunctions of linear systems.
  On the other hand, we give an algorithm for quantifier-free voting rules which is parameterized only by the number of conjunctions of the voting rule and runs in time polynomial in the number of types of people.
  This way, our framework explains why {\sc Shift Bribery} is polynomial-time solvable for the plurality voting rule, making explicit that the rule is simple in that it can be expressed with a single linear inequality, and that the number of voter types is polynomial.
\end{abstract}

\sloppy

\section{Introduction}
Problems of manipulation, bribery and control constitute a fundamental part of computational social choice.
Many such problems are known to be \NP-hard (or worse).
However, their input can naturally be partitioned into several parts, like the number of voters, the number of candidates, and others.
This motivates the study of such problems by the powerful tools of parameterized complexity.
One of the most fundamental parameters is the number of candidates $|C|$, which in many real-life scenarios can be expected to be reasonably small.

A by-now classical example in this direction is the $\mathcal R$-{\sc Swap Bribery} problem, which takes as input an election consisting of a~set $C$ of candidates and a~set $V$ of voters with their individual preference lists $\pref_v$ (for $v\in V$), which are total orders over~$C$.
Additionally, for each voter $v\in V$ and each pair of consecutive candidates $c \pref_v c'$, there is some cost $\sigma^v(c,c')\in\mathbb Z$ of swapping the order of $c$ and $c'$ in $\pref_v$.
The objective is to find a minimum-cost set of swaps of consecutive candidates in the preference lists in order to make a~designated candidate $c^\star\in C$ the winner of the thus-perturbed election under a fixed voting rule~$\mathcal R$.

This problem was introduced by Elkind et al.~\cite{ElkindEtAl2009} and has since been studied for many classical voting rules $\mathcal R$~\cite{DornSchlotter2012,FaliszewskiEtAl2014,KnopEtAl2017,SchlotterEtAl2017}.
In particular, its computational complexity has been thoroughly analyzed with respect to the number of candidates $|C|$.
The observation that $|C|$ is often small motivated the search for \emph{fixed-parameter algorithms} for $\mathcal R$-{\sc Swap Bribery} parameterized by $|C|$, which are algorithms that run in time $f(|C|)\cdot n^{O(1)}$ for some computable function $f$, here $n$ denotes the size of the input election; if such an algorithm exists, we then say that the problem is \emph{\fpt} with respect to the parameter $|C|$.

Despite the problem's importance, for a long time, only the ``uniform cost'' case of $\mathcal R$-{\sc Swap Bribery} was known to be \fpt for various voting rules, parameterized by the number of candidates; here, \emph{uniform cost} refers to the special case that all voters have the same cost function, that is, $\sigma^v \equiv \sigma$ for all $v\in V$.
This is a~fundamental result due to Dorn and Schlotter~\cite{DornSchlotter2012}, who showed that $\mathcal R$-{\sc Swap Bribery} with uniform cost can be solved in time $2^{2^{|C|^{\Oh(1)}}}\cdot n^{O(1)}$ for all voting rules $\mathcal R$ that are ``linearly describable''.
Many classical voting rules are indeed linearly describable, like any scoring protocol, Copeland$^\alpha$, Maximin, or Bucklin.

Recently, Knop et al.~\cite{KnopEtAl2017} gave the first fixed-parameter algorithms for $\mathcal R$-{\sc Swap Bribery} for \emph{general cost functions} for most voting rules $\mathcal{R}$ studied in the literature (scoring protocol, Copeland$^{\alpha}$, Maximin, Bucklin etc.), thereby removing the uniform cost assumption.
This way, they resolved a long-standing open problem.
Moreover, their algorithm runs in time $2^{|C|^{O(1)}}\cdot n^{O(1)}$ for many rules $\mathcal R$, and thus improves the double-exponential run time by Dorn and Schlotter.
Their key idea was to reduce the problem to so-called \emph{$n$\hy{}fold integer programming}, which allowed them to solve the problem efficiently for bounded number of candidates despite their integer program having an unbounded number of variables.
Their approach also solved $\mathcal R$-{\sc Swap Bribery} for~$\mathcal R$ being the Kemeny rule, even for general cost functions, though the Kemeny rule is not known to be linearly describable (cf.~\cite[p. 338]{FaliszewskiEtAl2011}).
However, this does not apply for Dodgson's and Young's rules.

\subsection{The challenge}
Even so, there are some notable voting rules $\mathcal R$ for which the complexity of $\RR$-{\sc Swap Bribery} remained open even in the uniform cost case.
This includes the Dodgson rule and the Young rule.
Those rules are based on the notion of \emph{Condorcet winner}, which is a candidate who beats any other candidate in a head-to-head contest.
The Condorcet voting rule is very natural and dates back to the 18\textsuperscript{th} century; however, clearly there exist elections without a Condorcet winner.
In such a situation one proclaims those candidates as winners who are ``closest'' to being a Condorcet winner; different notions of closeness then yield different voting rules:
\begin{itemize}[leftmargin=*]
  \item Closeness measured as the of number of swaps in voter's preference orders defines the \emph{Dodgson rule}.
  \item Closeness measured as the number of voter deletions defines the \emph{Young rule}.
\end{itemize}
Thus, a candidate $c$ is a Dodgson winner if s/he can be made a Condorcet winner by a minimum number of swaps in the voter's preference orders over all candidates; analogously for the Young rule and voter deletions.

Kemeny rule\footnote{Faliszewski~\cite{FaliszewskiEtAl2009} give a fixed-parameter algorithm for the $\mathcal R$-{\sc Bribery} problem when $\mathcal R = $Kemeny; $\mathcal R$-{\sc Bribery} is, in a sense, simpler than $\mathcal{R}$-\textsc{Swap Bribery} because in \textsc{Bribery} the cost of bribing a voter does not depend on \emph{how} we bribe, while in \textsc{Swap Bribery} the cost is the sum of costs for each performed swap.},

When considering $\mathcal{R}$-\textsc{Swap Bribery}, the Dodgson rule and the Young rule are much more complicated to handle than other rules; the reasons are several.
First, for many voting rules $\RR$, the winner of an election can be found in polynomial time, and solving this winner determination problem is certainly a necessary subtask when solving $\RR$\hy{}\textsc{Swap Bribery}.
However, for \mbox{$\RR \in \{\mbox{Dodgson, Young, Kemeny}\}$}, already winner determination is \NP{}\hy hard, and so even \emph{verifying} a solution (that $c^\star$ is indeed the winner of the perturbed election) is intractable.
However, for \mbox{$\RR \in \{\mbox{Dodgson, Young}\}$}, already winner determination is \NP{}\hy hard, and so even \emph{verifying} that~$c^\star$ is indeed the winner of the perturbed election is intractable.
In fact, winner determination for these voting rules is complete for parallel access to \NP~\cite{HemaspaandraEtAl2005,HemaspaandraHJ99,RotheSV03}, denoted $\mathsf{P}^{\NP}_{||}$\hy{}complete\footnote{The class $\mathsf{P}^{\NP}_{||}$ contains all problems solvable in polynomial time by a~deterministic Turing machine which has access to an \NP oracle, but must ask all of its oracle queries at once (i.e., the queries can not depend on each other).}.
Second, for more than 25 years the winner determination problem for the Dodgson rule and Young rule was only known to be solvable by an ILP-based algorithm~\cite{BartholdiEtAl1989} with doubly-exponential dependence in $|C|$; a~single-exponential algorithm is only known since recently~\cite{KnopEtAl2017}.
Even though winner determination for these rules turns out to be \fpt with parameter $|C|$, there is a~sharp difference: for the Kemeny rule, a~simple procedure enumerating all $|C|!$ possible preference orders suffices to determine the winner, while for Dodgson and Young, only a~double-exponential ILP-based algorithm was known for a~long time~\cite{BartholdiEtAl1989} and a~single-exponential algorithm is only known recently~\cite{KnopEtAl2017}.
This provides a sharp contrast to the Kemeny rule, for which simply enumerating all $|C|!$ possible preference orders suffices to determine the winner.
Faliszewski et al.~\cite{FaliszewskiEtAl2009} describe these difficulties:
\newenvironment{myquote}[1]%
{\list{}{\leftmargin=#1\rightmargin=#1}\item[]}%
{\endlist}
\begin{myquote}{0.1in}
  \emph{It is interesting to consider which features of Kemeny elections allow us to employ the above [ILP-based] attack, given that the same approach does not seem to work for either Dodgson or Young elections.
  One of the reasons is that the universal quantification implicit in Dodgson and Young elections is over an exponentially large search space, but the quantification in Kemeny is, in the case of a fixed candidate set, over a fixed number of options.}
\end{myquote}
Thus, it is not clear how to solve $\mathcal R$-{\sc Swap Bribery} even for uniform cost with \emph{any} fixed-parameter algorithm for $\mathcal R$ being the Dodgson rule or the Young rule.
These complications led Faliszewski et al.~\cite{FaliszewskiEtAl2011} to explicitly ask for the complexity of $\RR$-{\sc Swap Bribery} parameterized by the number of candidates under these rules.

\subsection{Our contributions}
\label{sec:ourcontributions}
We start by making a key observation about the majority of fixed-parameter algorithms for $\RR$-{\sc Swap Bribery} when $|C|$ is small.
A typical such result is an algorithm for $\mathcal R$-\textsc{Swap Bribery} for $\mathcal R$ being Con\-dorcet's voting rule.
That algorithm uses two key ingredients:
\begin{enumerate}[leftmargin=*]
  \item There are at most $|C|!$ preference orders of $C$, and hence each voter falls into one of $|C|!$ types; thus, an input election is expressible as a society $\ves = (s_1, \dots, s_{|C|!})$, where $s_i$ is the number of voters of type $i$.
  \item Expressing that a candidate $c^\star$ is a Condorcet winner is possible using a conjunction of $|C|-1$ linear inequalities in terms of $\ves$.
\end{enumerate}
As those key properties hold almost universally for voting rules~$\mathcal R$, one might be tempted to think that if there are many types of voters, the $\RR$-{\sc Swap Bribery} problem must be hard, and if there are few types of voters, the problem must be easy.
However, two points arise as counter-evidence.
First, very recently, Knop et al.~\cite{KnopEtAl2017} showed that even if there are many types of voters who differ by their cost functions, the $\RR$-\textsc{Swap Bribery} problem remains fixed-parameter tractable for a wide variety of voting rules $\RR$.
Second, as already mentioned, it was open since 2007 whether $\mathcal R$-\textsc{Swap Bribery} with Dodgson's and Young's voting rule are fixed-parameter tractable for few candidates, even for uniform cost functions.

\heading{From voters and candidates to societies}
Here, we take a novel perspective.
We observe that the two key ingredients (1) and (2) apply much more widely than for $\RR$-\textsc{Swap Bribery}; namely, they are also present in many other manipulation, bribery and control problems.
We therefore abstract away the specifics of such problems and introduce general notions of ``society'', ``moves in societies'', and ``winning conditions''.
Let $\tau \in \N$ be the number of types of people (e.g., voters in an election or a referendum).
A \emph{society} $\ves$ is simply a non-negative $\tau$-dimensional integer vector encoding the numbers of people of each type.
A \emph{move} $\vem$ is a $\tau^2$-dimensional integer vector whose elements sum up to zero; it encodes how many people move from one type to another.
A \emph{change} $\veDelta$ is a $\tau$-dimensional vector (typically associated with a move) encoding the effect of a move on a society, such that $\ves + \veDelta$ is again a society.
Finally, a \emph{winning condition} $\Psi(\ves)$ is a predicate encoding some desirable property of a society, such as that a preferred candidate has won or that a preferred agenda was selected in a referendum.
Specifically, we study winning conditions which are describable by formulas in \emph{Presburger Arithmetic} (PA).
PA is a logical language whose atomic formulas are linear inequalities over the integers, which are then joined with logical connectives and quantifiers.
Thus, winning conditions describable by PA formulas widely generalize the class of linearly describable voting rules by Dorn and Schlotter~\cite{DornSchlotter2012}.

Our main technical contribution informally reads as follows:
\begin{theorem}[informal]
\label{thm:main_thm}
Deciding satisfiability of PA formulas with two quantifiers is fixed-parameter tractable with respect to the dimension and length of formula, provided its coefficients and constants are given in unary.
\end{theorem}

The importance of Theorem~\ref{thm:main_thm} arises from its applicability to the following general manipulation problem that we introduce here.
This general manipulation problem, which we call {\sc Minimum Move}, captures that many manipulation problems can be cast as finding a minimum move with respect to some objective function; in particular, it encompasses the well-studied $\RR$-{\sc Swap Bribery} problem.
We study {\sc Minimum Move} for linear objective functions and winning conditions~$\Psi$ expressible with PA formulas of the form~``$\exists\forall$''.
For all such~$\Psi$, with the help of Theorem~\ref{thm:main_thm}, we show that {\sc Minimum Move} is fixed-parameter tractable for combined parameter the descriptive complexity (length) of the winning condition~$\Psi$ and the number $\tau$ of ``types of people'', that is, it is fixed-parameter tractable for parameter lengthe of $\Psi$ plus $\tau$.
As an important special case, we obtain the first fixed-parameter algorithm for $\mathcal R$-{\sc Swap Bribery} for~$\RR$ the Dodgson rule and the Young Rule with uniform costs.
To this end, we model the winning condition of the Dodgson rule and Young rules as a PA formula.
For intuition, consider the Young rule: a~candidate $c^{\star}$ is a Young winner (with score $d$) if there \emph{exists} a set $V^\star \subseteq V$ of at most $d$ voters such that $c^{\star}$ is a~Condorcet winner of the election $(C, V \setminus V^\star)$, and \emph{for all} sets $V' \subseteq V$ of at most $d-1$ voters any other candidate $c \neq c^{\star}$ is not a~Condorcet winner of the election $(C, V \setminus V')$.
This formula has one quantifier alternation, and its length (for a fixed score $d$) is bounded by some function of $|C|$; finally we have to take a disjunction of such formulas over all possible scores $d$.
For a candidate set $C$, the number $\tau$ of types of people is bounded by $|C|!$.
Consequently, we finally settle the long-standing open question about the complexity of $\RR$-\textsc{Swap Bribery} for $\RR \in \{\textrm{Dodgson, Young}\}$, that was explicitly raised by Faliszewski~\cite{FaliszewskiEtAl2011}:
\begin{theorem}
\label{thm:main-fpt-swapbribery}
  $\RR$-{\sc Swap Bribery} with uniform cost is \fpt parameterized by the number of candidates for~$\RR$ being the Dodgson rule or the Young rule; it can be solved in time $f(|C|) \cdot |V|^{O(1)}$ for some computable function $f$.
\end{theorem}
Beyond this fundamental problem, we show that a host of other well-studied manipulation problems are captured by our fixed-parameter algorithm for {\sc Minimum Move}:
\begin{corollary}
\label{cor:main-fpt-other}
  For $\RR \in \{\textrm{Dodgson, Young}\}$, the following problems are fixed-parameter tractable for uniform costs when parameterized by the number $|C|$ of candidates: $\RR$-\textsc{\$Bribery}, $\RR$-\textsc{CCDV/CCAV}, $\RR$-\textsc{Possible Winner}, and $\RR$-\textsc{Extension Bribery}.
\end{corollary}

Let us turn our attention to the parameter ``number of types of people'' $\tau$.
Our main contribution here is the following:
\begin{theorem}[informal]
\label{thm:polymanytypes-informal}
For any quantifier-free winning condition $\Psi$, {\sc Minimum Move} can be solved in time \emph{polynomial} in the number of types and exponential only in the number of linear inequalities of\/ $\Psi$.
\end{theorem}
Note that in many models of bribery and control, the number of \emph{potential} types of people (i.e., types that can occur in any feasible solution) is polynomial in the number of people on input.
For example, in \textsc{Shift Bribery}, every voter can be bribed to change their preferences order to one of $|C|-1$ orders; thus the number of potential types is $(|C|-1)|V|$.
Similarly, in \textsc{CCAV / CCDV} (constructive control by adding or deleting voters), every voter has an active/latent bit; thus the number of potential types is $2|V|$.
Similar arguments also work for \textsc{Support Bribery} where we change voters' approval counts, and with a more intricate argumentation also for some voting rules and \textsc{Bribery} and \textsc{\$Bribery}.
In this sense, the fact that we need to consider $|C|!$ potential voter types in $\RR$-\textsc{Swap Bribery} almost seems like an anomaly, rather than a rule.
In summary, the complexity of {\sc Minimum Move} depends primarily on the descriptive complexity of the winning condition $\Psi$, because in many cases the number of types of people is polynomially bounded.

Another consequence of Theorem~\ref{thm:main_thm} are the first fixed-parameter algorithms for two important manipulation problems beyond $\RR$-{\sc Swap Bribery}.
The {\sc Resilient Budget} problem asks, for a given society whether allocating budget $B$ is sufficient in order to repel any adversary move of cost at most $B_a$ with a counter-move of cost at most $B$ (so that the winning condition is still satisfied).
Similarly, \textsc{Robust Move} asks for a move of cost at most $B$ which causes the winning condition to be satisfied even after any adversary move of cost at most $B_a$.
For formal definitions and results, cf. Sect.~\ref{subsec:robustresilient}.

\subsection{Interpretation of results}
Intuitively, the results obtained with Theorem~\ref{thm:main_thm} can be interpreted as follows.
Dodgson-\textsc{Swap Bribery} is fixed-parameter tractable parameterized by $|C|$; however, this comes with at least two limitations as compared to prior work for simpler voting rules~$\RR$.
First, our methods do not extend beyond the uniform cost scenario, and this remains a major open problem.
Second, our result requires the input election to be given in unary, while prior work allows it to be given in binary (this is sometimes called the \emph{succinct} case~\cite{FaliszewskiEtAl2009}).
This is easily explained by the different descriptive complexities of the respective voting rules: for example, while Condorcet's voting rule can be formulated as a quantifier-free PA formula, formulating Dodgson's rule requires a long disjunction of formulas which use two quantifiers and a bounded number of disjunctions.

Theorem~\ref{thm:polymanytypes-informal} lets us discuss more specifically the complexity of various voting rules.
For example, the Plurality voting rule can be expressed with a single linear inequality encoding that a preferred candidate obtained more points than the remaining candidates altogether.
Thus, all problems which can be modeled as {\sc Minimum Move} are polynomial-time solvable with the Plurality voting rule.
This interprets the result of Elkind et al.~\cite[Theorem 4.1]{ElkindEtAl2009} that Plurality-\textsc{Shift Bribery} is polynomial-time solvable: the number of potential voter types is polynomial, and Plurality has a simple description.
Continuing, we may compare $\RR$=Borda with $\RR$=Copeland.
The winning condition for $\RR$=Borda can be described with $|C|-1$ inequalities, while $\RR$=Copeland requires $\Oh(|C|^2)$ inequalities.
Thus Borda-\textsc{Swap Bribery} is solvable in time $|C|^{\Oh(|C|^2)} \log |V|$, while Copeland-\textsc{Swap Bribery} requires time $|C|^{\Oh(|C|^4)} \log |V|$.
Finally, all descriptions of Kemeny's voting rule we are aware of require $|C|!$ inequalities, and thus result in Kemeny-\textsc{Swap Bribery} being solvable in time $|C|!^{(|C|!)^{\Oh(1)}} \log |V|$.
We do not claim these complexities to be best possible, but conjecture the existence of lower bounds separating the various voting rules; in particular, we believe that Kemeny-\textsc{Swap Bribery} requires double-exponential time.

Finally, our work provides a natural next step in unifying the many different models that have been proposed for voting, bribing and manipulation problems.
In this direction, Faliszewski et al.~\cite{FaliszewskiEtAl2011} study what happens when multiple bribery and manipulation actions can occur in an election; e.g., CCAV asks for constructive control by adding voters while CCDV by deleting voters; similarly for CCAC and CCDC for adding/deleting candidates.
Faliszewski et al. unify those various (up to that point separately studied) attacks.
Similarly, Knop et al.~\cite{KnopEtAl2017} formulate the $\RR$-{\sc Multi Bribery} problem, which also incorporates swaps and perturbing approval counts.
The problem we put forward in this paper, {\sc Minimum Move}, in some sense generalizes and simplifies all those ``meta''-problems.

\subsection{Related work}
\label{sec:relatedwork}
We have reviewed most of the relevant computational social choice work already.
However, there seems to be some confusion in the literature that deserves clarification.
The paper of Faliszewski et al.~\cite{FaliszewskiEtAl2009} pioneering the concept of bribery in elections indeed considers the voting rules Kemeny, Dodgson and Young, and provides a fixed-parameter algorithm for Kemeny-\textsc{Bribery}.
There are three features of their paper that we wish to discuss.

First, turning their attention to Dodgson-\textsc{Bribery}, they write:
\begin{myquote}{0.1in}
  \emph{Applying the integer programming attack for the case of bribery within Dodgson-like election systems [...] is more complicated.
  These systems involve a more intricate interaction between bribing the voters and then changing their preferences.
  For Dodgson elections, after the bribery, we still need to worry about the adjacent switches within voters’ preference lists that make a particular candidate a Condorcet winner. [...]
  This interaction seems to be too complicated to be captured by an
  integer linear program, but building on the flavor of the Bartholdi et al.~\cite{BartholdiEtAl1989} ILP attack we can achieve the following: Instead of making $p$ a winner, we can attempt to make $p$ have at most a given Dodgson or Young score.}
\end{myquote}
They call this problem DodgsonScore-\textsc{Bribery} and provide positive results for it.
Notice, however, that finding a bribery which makes~$c^{\star}$ have a certain Dodgson score does \emph{not} prevent another candidate to have a lower score and winning the bribed election.
Thus, solving DodgsonScore-\textsc{Bribery} can be very far from the desired result.

Second, the authors then observe that a brute force approach enumerating all $|V|^{|C|!}$ briberies solves the Dodgson-\textsc{Bribery} problem in polynomial time for \emph{constantly} many candidates; however, theirs is \emph{not} a fixed-parameter algorithm for parameter $|C|$.

Third, they then introduce another voting system called Dodgson$'$, which is similar to Dodgson, and provide a fixed-parameter algorithm for winner determination.
However, as in the case of Dodgson-\textsc{Bribery}, they do not provide a fixed-parameter algorithm for Dodgson$'$-\textsc{Bribery}.

The issue is then that a subsequent paper of Falisezwski et al.~\cite{FaliszewskiEtAl2011} claims that the Dodgson rule is ``integer-linear-program implementable'' and that this implies a certain election control problem generalizing \textsc{Bribery} to be \fpt~\cite[Theorem 6.2]{FaliszewskiEtAl2011}.
We believe the authors do not sufficiently differentiate between determining the winner with \emph{one} ILP, as is the case for most simple voting rules, and with \emph{multiple} ILPs, as is the case for Dodgson.
Thus, we believe there is no evidence that the Dodgson rule is ``integer-linear-program implementable''.
Yet, this may be possible and this question still deserves attention.
Whatever the reason, we are convinced that their~\cite[Theorem 6.1]{FaliszewskiEtAl2011} does \emph{not} hold for $\RR$=Dodgson.
Hence, we believe that ours are the first fixed-parameter algorithms for any \textsc{Bribery}\hy like problem for $\RR \in$~\{Dodgson, Young\}.

\section{Preliminaries}
\label{sec:preliminaries}
Let $m,n$ be integers.
We define \mbox{$[m,n] := \{m, m+1, \ldots, n\}$} and \mbox{$[n] := [1,n]$}.
Throughout, we reserve bold face letters (e.g. $\vex,\vey$) for vectors.
For a~vector $\vex$ its $i$-th coordinate is $x_i$.

Next, we provide notions and notations for $\RR$-\textsc{Swap Bribery}.

\heading{Elections}
An election~$(C,V)$ consists of a set $C$ of candidates and a set~$V$ of voters, who indicate their preferences over the candidates in $C$, represented via a \emph{preference order} $\pref_v$ which is a total order over~$C$.
We often identify a voter $v$ with their preference order~$\pref_v$.
Denote by $\textrm{rank}(c,v)$ the rank of candidate~$c$ in~$\pref_v$; $v$'s most preferred candidate has rank 1 and their least preferred candidate has rank $|C|$.
For distinct candidates~$c,c'\in C$, write $c\pref_v c'$ if voter~$v$ prefers~$c$ over~$c'$.

\heading{Swaps}
Let $(C,V)$ be an election and let $\pref_v\in V$ be a voter.
For candidates $c,c'\in C$, a \emph{swap} $s = (c,c')_v$ means to exchange the positions of $c$ and $c'$ in~$\pref_v$; denote the perturbed order by $\pref_v^s$.
A swap~$(c,c')_v$ is \emph{admissible in $\pref_v$} if $\rank(c,v) = \rank(c',v)-1$.
A set $S$ of swaps is \emph{admissible in $\pref_v$} if they can be applied sequentially in~$\pref_v$, one after the other, in some order, such that each one of them is admissible.
Note that the perturbed vote, denoted by $\pref_v^S$, is independent from the order in which the swaps of $S$ are applied.
We extend this notation for applying swaps in several votes and denote it $V^S$.
We specify $v$'s cost of swaps by a function $\sigma^v: C\times C\rightarrow \mathbb{Z}$.

\heading{Voting rules}
A voting rule~$\RR$ is a function that maps an election $(C,V)$ to a subset $W\subseteq C$, called the \emph{winners}.
A candidate~$c\in C$ is a \emph{Condorcet winner} if any other~$c'\in C \setminus \{c\}$ satisfies $\big|\{\pref_v \in V \mid c\pref_v c' \}\big| > \big|\{\pref_v \in V \mid c' \pref_v c\}\big|$; then we say that \emph{$c$ beats $c'$ in a head-to-head contest}.
The \emph{Young score of $c \in C$} is the size of the smallest subset $V' \subseteq V$ such that $c$ is a Condorcet winner in $(C, V \setminus V')$.
Analogously, the \emph{Dodgson score of $c \in C$} is the size of the smallest admissible set of swaps $S$ such that $c$ is a Condorcet winner in $(C, V^S)$.
Then, $c$ is a \emph{Young (Dodgson) winner} if it has minimum Young (Dodgson) score.

We aim to solve the following problem:
\prob{$\RR$-\textsc{Swap Bribery}}
{An election $(C,V)$, a designated candidate $c^{\star} \in C$ and swap costs $\sigma^v: C\times C\rightarrow \mathbb{Z}$ for $v \in V$.}
{A set $S$ of admissible swaps of minimum cost so that $c^{\star}$ wins the election $(C, V^S)$ under the rule~$\RR$.}

\section{Moves in Societies and Presburger Arithmetic}
Let $\tau \in \N$ be the number of \emph{types of people}.

\begin{definition}
  A \emph{society} is a non-negative $\tau$-dimensional integer vector $\ves = (s_1, \dots, s_\tau)$.
\end{definition}

In most problems, we are interested in modifying a society by moving people between types.

\begin{definition}
  A \emph{move} is a vector $\vem = (m_{1,1}, \dots, m_{\tau, \tau})\in\mathbb Z^{\tau^2}$.
\end{definition}
\noindent
Intuitively, $m_{i,j}$ is the number of people of type $i$ turning type~$j$.
\begin{definition}
  A \emph{change} is a vector $\veDelta = (\Delta_1, \dots, \Delta_\tau)\in\mathbb Z^{\tau}$ whose elements sum up to $0$.
  We say that \emph{$\veDelta$ is the change associated with a move $\vem$} if $\Delta_i = \sum_{j=1}^\tau m_{j,i} - m_{i,j}$, and we write $\veDelta = \Delta(\vem)$.
  A change $\veDelta$ is \emph{feasible with respect to society $\ves$} if $\ves + \veDelta \geq \mathbf{0}$, i.e., if applying the change $\veDelta$ to $\ves$ results in a society.
\end{definition}

One more useful notion is that of a move costs vector:
\begin{definition}
  A \emph{move costs vector} is a vector $\vecc = (c_{1,1}, \dots, c_{\tau, \tau})$ in $(\N \cup \{+\infty\})^{\tau^2}$ which satisfies the triangle inequality, i.e., $c_{i,k} \leq c_{i,j} + c_{j,k}$ for all distinct $i,j,k$.
\end{definition}

\begin{definition}[$(\vecc, k)$-move]
Let $k \in \N$ and $\vecc$ be a move costs vector.
A move $\vem$ is a \emph{$(\vecc, k)$-move} if $\vecc^\transpose \vem \leq k$.
\end{definition}

Finally, we want to check that the society (e.g., resulting from applying some moves) satisfies a certain desired condition.
This condition depends on the problem we are modeling: in variants of bribery, it says that a preferred candidate is elected as a winner or to be a part of a committee under a given voting rule; in the context of lobbying, it says that a preferred agenda was selected.
To allow large expressibility, we make a very broad definition:
\begin{definition}
    A \emph{winning condition of width $\tau$} is a predicate $\Psi(\ves)$ with $\tau$ free variables.
\end{definition}

\subsection{Presburger Arithmetic}
For two formulas $\Phi$ and $\Psi$, we write $\Phi \cong \Psi$ to denote their equivalence.

\begin{definition}[Presburger Arithmetic]
  Let
  \[
    \hat{\P}_{0,(n_0),\delta, \gamma,\alpha, \beta} = \left\{\hat{\Psi}(\vex_0) \right\}
  \]
  be the set of quantifier-free Presburger Arithmetic (PA) formulas with $n_0$ free variables $\vex_0$ which are a disjunction of at most $\delta$ conjunctions of linear inequalities $\vea^{\transpose}\vex_0 \leq b$, each of length at most $\gamma$, where $\|\vea\|_\infty \leq \alpha$ and $|b| \leq \beta$ for each inequality.
  Then, let
  \[
    \P_{0,(n_0),\delta, \gamma,\alpha, \beta} = \left\{\Psi(\vex_0) \mid \exists \hat{\Psi}(\vex_0) \in \hat{\P}_{0, (n_0), \delta, \gamma, \alpha, \beta}: \hat{\Psi} \cong \Psi \right\}
  \]
  be the set of PA formulas equivalent to some DNF formula from $\hat{\P}_{0, (n_0), \delta, \gamma, \alpha, \beta}$.
  Finally, let
  \[
    \P_{k,\ven,\delta, \gamma,\alpha, \beta} = \left\{\Psi(\vex_0) \equiv \exists / \forall \vex_1 \exists / \forall \vex_2 \cdots \exists / \forall \vex_k: \Phi(\vex_0, \vex_1, \dots, \vex_k)\right\}
  \]
  be the set of PA formulas with \emph{quantifier depth} $k$, $n_0$ \emph{free variables} $\vex_0$, and \emph{dimension} $\ven = (n_0, n_1, \dots, n_k)$; here, $x_i \in \Z^{n_i}$ for each $i$ and $\Phi(\vex_0, \dots, \vex_k) \in \P_{0, (n), \delta, \gamma, \alpha, \beta}$ with $n = \sum_{i=1}^k n_i$.
  The \emph{length} of $\Psi(\vex_0)$ is the number of symbols it contains, which is polynomially bounded in $(n,\gamma,\delta)$.
  By $\exists \P$ and $\forall \P$ we denote the sets of PA formulas whose leading quantifier is $\exists$ or $\forall$, respectively.
\end{definition}

\begin{example*}
  A simple example of PA is the following formula.
  \begin{align*}
    \Psi(y) \equiv \forall x_1 x_2 \exists z_1 z_2 z_3 & ~~:~~  \left(x_1 + y = z_3 \land \ y\ge 0 \right) \lor \\
			& \big(3x_1 + 10y - 3z_1 \le 13 \land 2x_2 + 5y - z_2 \le 11 \\
			& \land x_1 + 1y - z_3 \ge 9 \land z_1 -z_2 + 2z_3 \le 6 \big)
  \end{align*}
  Here $k = 2, \ven = \left(1,2,3\right), \delta=2, \gamma=4, \alpha=10$, and $\beta=13$.
\end{example*}
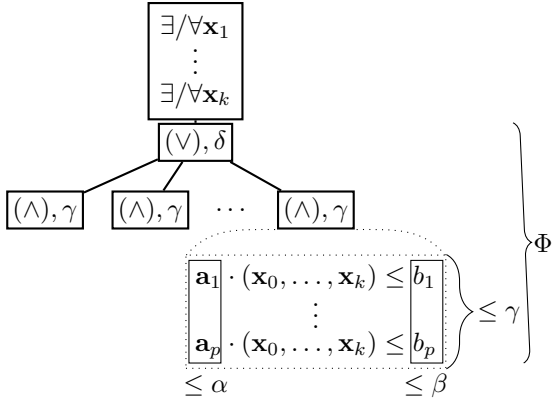
\begin{figure}[bt]
  \usetikzlibrary{calc,positioning,decorations.pathreplacing,fit}

\begin{tikzpicture}
\tikzstyle{edge}=[thick]
\tikzstyle{branchPoint}=[draw, thick, inner sep=2pt]
\tikzstyle{usual}=[inner sep=2pt]
\tikzstyle{nafouknuti}=[draw, dotted, inner sep=2pt]

\newcommand{\vdist}{.9}

\node[usual] (x1) {$\exists / \forall \vex_1$};
\node[usual] (xk) at ($(x1) - (0,\vdist)$) {$\exists / \forall \vex_k$};
\node[usual] (xes) at ($(x1.south)!.5!(xk.north) + (0,.1)$) {\vdots};
\node[branchPoint] at ($(x1) - (0,1.7*\vdist)$) (Disj) {$(\lor), \delta$};
\node[branchPoint, fit = (x1)(xk)(xes)] (x) {};

\draw[edge] (x) -- (Disj);

\node[branchPoint] (land1) at ($(Disj) - (2,  \vdist)$) {$(\land), \gamma$};
\node[branchPoint] (land2) at ($(Disj) - (.6, \vdist)$) {$(\land), \gamma$};
\node[branchPoint] (landDelta) at ($(Disj) - (0, \vdist) + (1.6, 0)$) {$(\land), \gamma$};
\node at ($(land2)!.5!(landDelta)$) {\dots};

\node at ($(landDelta) - (0,\vdist)$) (a1) {$\vea_1\cdot(\vex_0,\ldots,\vex_k) \le b_1$};
\node at ($(a1) - (0,\vdist)$) (ap) {$\vea_p\cdot(\vex_0,\ldots,\vex_k) \le b_p$};
\node at ($(a1)!.5!(ap) + (0,.05)$) {\vdots};

\node[nafouknuti,fit=(a1)(ap),inner sep=0pt] (fitA) {};
\draw[dotted] (a1.north west) to[in=180, out=50] (landDelta.south west);
\draw[dotted] (a1.north east) to[in=0, out=130] (landDelta.south east);
\draw[decorate,decoration={brace,amplitude=10pt}] (fitA.north east) -- (fitA.south east) node [midway,xshift=0.7cm]  {$\le \gamma$};

\coordinate (ap_a) at ($(ap.south west) + (.5,0) + (-1pt,2pt)$);
\coordinate (ap_b) at ($(ap.south west) + (1pt,2pt)$);
\draw[black] ($(a1.north west) + (1pt, -2pt)$) rectangle  (ap_a);
\node at ($(ap_a.south) - (.2,.3)$) {$\le\alpha$};

\coordinate (bp_a) at ($(ap.south east) - (.5,0) + (1pt,2pt)$);
\draw[black] ($(a1.north east) - (1pt, 2pt)$) rectangle  (bp_a);
\node at ($(bp_a.south) - (-.2,.3)$) {$\le\beta$};

\foreach \dest in {land1, land2, landDelta}
  \draw[edge] (Disj) -- (\dest);

\draw[decorate,decoration={brace,amplitude=5pt}] ($(Disj.north east) + (3.72cm,0cm)$) -- ($(bp_a) + (1.5cm,0)$) node [midway,xshift=0.33cm]  {$\Phi$};

\end{tikzpicture}
  \caption{Visualization of PA formula $\Psi(\vex_0)$ in DNF.}
\end{figure}

We study winning conditions $\Psi(\ves)$ expressible in PA, and state our complexity results with respect to the \emph{descriptive complexity} of~$\Psi$, which is its number of variables, quantifiers, logical connectives, and unary encoding length of coefficients and constants.

\heading{Vocabulary}
We express relevant definitions by simple PA formulas over integral variables and with integer coefficients and constants:
\begin{itemize}
    \item $\texttt{society}(\ves) \equiv \ves \geq \mathbf{0} \in \P_{0, (\tau), 0, \tau, 1, 0}$,
    \item $\texttt{move}(\vem)  \in \P_{0, (\tau^2), 0, 0, 0, 0}$,
    \item $\veDelta = \Delta(\vem)$ is a linear map $\Delta_i = \sum_{j=1}^\tau m_{j,i} - m_{i,j}$; thus if we let $\Psi(\vem, \veDelta) \equiv \veDelta = \Delta(\vem)$, then $\Psi(\vem, \veDelta) \in \P_{0, (\tau + \tau^2), 0, \tau, 1, 0}$
    \item$\texttt{feasible}(\ves, \veDelta) \equiv \left( \ves + \Delta(\vem) \geq \mathbf{0} \wedge \mathbf{1}^\intercal \veDelta = 0 \right) \in \P_{0, (2\tau), 0, \tau, 2, 0}$, and,
    \item$(\vecc, k)$-$\texttt{move}(\vem) \equiv \vecc^\transpose \vem \leq k \in \P_{0, (\tau), 0, 1, \|\vecc\|_\infty, k}$.
\end{itemize}

We note that, for elections, our definition of winning condition generalizes the notion of ``linearly-definable voting rules'' by Dorn and Schlotter~\cite{DornSchlotter2012}.
Precisely, those rules belong to $\exists \P$ with $k=1$; we will show that Dodgson and Young are in $\exists \P$ with $k=2$.
Thus, our winning conditions capture an extensive set of voting rules.

\subsection{Modeling Problems as Minimum Move}
We model moves in societies by the following general problem:
\prob{{\sc Minimum Move}}
{A society $\ves$, an objective function $f: \Z^{\tau^2} \to \Z$, a winning condition $\Psi$.}
{A move $\vem$ minimizing $f(\vem)$ s.t. $\Psi(\ves + \Delta(\vem)) \wedge \texttt{feasible}(\ves, \Delta(\vem))$.}

It models many well-studied problems:

\heading{Multi bribery}
Knop et al.~\cite{KnopEtAl2017} introduce a generalization of various bribery problems called $\RR$-\textsc{Multi Bribery}.
Informally, we are given an election $(C,V)$ where each voter further has an approval count, and is either active or latent the status of which can be changed at certain cost; likewise, there are costs for perturbing their preference order or approval count.
This problem generalizes \textsc{Bribery}, \textsc{\$Bribery}, \textsc{Swap Bribery}, \textsc{Shift Bribery}, \textsc{Support Bribery}, \textsc{Extension Bribery}, \textsc{Possible Winner}, \textsc{Constructive Control by Adding/Deleting Voters} and other problems.

Notice that there are at most $|C|!$ possible preference orders, at most $|C|$ possible approval counts, and $2$ states ``active'' or ``latent''.
Thus, there are at most $\tau \leq 2 |C|\cdot |C|!$ potential types of voters, and we can express the input election as a society $\ves$.
A move costs vector~$\vecc$ describing the costs of moving a voter from one type to another is obtained by calculating (possibly using a shortest path algorithm) the least costs based on the given cost functions.
Let $\Psi(\ves)$ be a PA formula which is satisfied if the preferred candidate wins under the voting rule $\RR$ in a society $\ves$.
Then, a bribery of minimum cost in a $\RR$-\textsc{Multi Bribery} instance can be modeled as solving {\sc Minimum Move} with $f(\vem) = \vecc^{\transpose} \vem$.
This modeling, combined with Theorem~\ref{thm:main-fpt-swapbribery} and Corollary~\ref{cor:main_cor}, yields Corollary~\ref{cor:main-fpt-other}.

\heading{Multiwinner elections}
Bredereck et al.~\cite{BredereckEtAl2016b} study the complexity of \textsc{Shift Bribery} in committee elections, that is, in elections with multiple winners.
The modeling is exactly the same as above, except for the winning condition $\Psi$ which will be a long disjunction over all committees which include the preferred candidate.

\heading{Lobbying in referenda}
Bredereck et al.~\cite{BredereckEtAl2014c} study the complexity of \textsc{Lobbying} in referenda.
There, voters cast ballots with their ``yes''/``no'' answers to issues.
The task is to push an agenda, i.e., a certain outcome.
Again, voters fall into groups according to their ballots, the costs of changing their opinions forms a move costs vector, and a winning condition $\Psi$ expresses that the selected agenda succeeded.

\section{Sentences With Two Quantifiers}
We shall now introduce the building blocks of our proof of Theorem~\ref{thm:main_thm}.
Woods~\cite{Woods2015} gives an algorithm that efficiently converts any quanti\-fier-free PA formula $\Phi$ into an equivalent DNF formula $\hat{\Phi}$ of bounded length:
\begin{lemma}[{Woods~\cite[Proposition 5.1]{Woods2015}}]
\label{lem:woods_dnf}
  Let~$\Phi(\vex)$ be a quanti\-fier-free PA formula with $\vex \in \Z^d$ containing $N$ inequalities, whose coefficients and right-hand sides are bounded in absolute value by $\alpha$ and $\beta$, respectively.
  Then $\Phi(\vex)$ can be converted into an equivalent DNF formula $\hat{\Phi}(\vex)$ with at most $\delta = N^{\Oh(d)}$ disjunctions, each containing at most $N$ conjunctions with the same bound on $\alpha$ and $\beta$.
\end{lemma}

It is often useful for the quantifiers of a PA formula to range over integer points of polyhedra, e.g. $\forall \vex \in Q$ (we do not write $Q \cap \Z^{n}$ for brevity, as we assume everything to be integer); again, our definition is not restrictive by the fact that we can always rearrange:
\begin{align*}
  \Psi(\vex_0) \equiv \exists \vex_1 \in Q_1 \cdots \forall/\exists \vex_k \in Q_k: &\, \Phi(\vex_0, \vex_1, \dots, \vex_k) \equiv \\
  \exists \vex_1 \cdots \forall/\exists \vex_k: & \, \Big(\Phi(\vex_0, \vex_1, \dots, \vex_k) \\ \bigwedge (\vex_1 \in Q_1 \wedge \vex_3 \in Q_3 \cdots) \wedge &
  \bigvee (\vex_2 \not\in Q_2 \vee \vex_4 \not\in Q_4 \cdots) \Big)
\end{align*}

\heading{Parametric ILP}
A special case of PA are \emph{parametric ILPs}, which can be viewed\footnote{Parametric ILPs are typically viewed as ILPs with a varying right hand side, that is, deciding the sentence $\forall \veb \exists \vex: A \vex \leq \veb$; it is known that our formulation is equivalent, as shown by Crampton et al.~\cite{CramptonEtAl2017}, who call it \textit{ILP Resiliency}.} as deciding the sentence
\[
  \forall \vex \in \Z^p: A \vex \leq \veb \, \exists \vey \in \Z^n: B(\vex, \vey) \leq \vece,
\]
where $A \in \Z^{\ell \times p}$ and $B \in \Z^{m \times n}$ are integer matrices.
A consequence of an algorithm of Eisenbrand and Shmonin~\cite{EisenbrandShmonin2008} is the following:
\begin{corollary}[{\cite[Theorem 4.2]{EisenbrandShmonin2008}, \cite[Corollary 1]{CramptonEtAl2017}}]
\label{cor:es_pilp}
  Any parametric ILP whose entries of $A, B$, $\veb$ and $\vece$ are given in unary, is fixed-parameter tractable when parameterized by $n,m$ and $p$.
\end{corollary}

\heading{ILP and disjunctions}
We shall use a folklore result about implementing disjunctions in ILP when the domains of variables can be bounded.
For that, we need another definition.

\begin{definition}[$B$-bounded, $B$-small PA formula]
  Let $\Psi(\vex_0) \equiv \exists / \forall \vex_1 \cdots \exists / \forall \vex_k: \Phi(\vex_0, \dots, \vex_k)$ be a PA formula, and let $\Psi_B(\vex_0) \equiv \exists / \forall \vex_1 \in [-B,B]^{n_1} \cdots \exists / \forall \vex_k \in [-B,B]^{n_k}: \Phi(\vex_0, \dots, \vex_k)$.
  Then we say that $\Psi(\vex_0)$ is~\emph{$B$-bounded} if
  \[
    \left\{\vex \in \Z^{n_0} \mid \Psi_B(\vex) \right\} \cap [-B,B]^{n_0} = \left\{ \vex \in \Z^{n_0} \mid \Psi(\vex) \right\},
  \]
  i.e., the set of feasible solutions does not change by restricting all quantifiers and free variables to the corresponding box of size $B$.

  Moreover, we say that any $\Psi \in \P_{k, \ven, \delta, \gamma, B, B}$ is~\emph{$B$\hy small} if it is $B$-bounded, that is, its coefficients and constants are bounded by~$B$.
\end{definition}

A special case are ILPs which are $B$-small; they correspond to PA formulas with $k=0$ and $\delta = 0$; for such formulas we show:
\begin{lemma}[{ILP disjunctions [folklore]}]
\label{lem:disjunctions}
  Let $A_i \vex \leq \veb_i$ for $i \in [d]$ be $B$\hy small ILPs with $A_i \in \Z^{m \times n}$ for each $i \in [d]$.
  Then, a~$(B^2 n)$\hy small system $A \vex \leq \veb$ with $A \in \Z^{(md+d+1) \times (n + d)}$ can be constructed in time $\Oh\left(dm + n + \sum_{i=1}^d \l A_i, b_i \r\right)$ such that
  \[
    \exists (\vex, \vey) \in \Z^{n+d}: A (\vex, \vey) \leq \veb \quad \Longleftrightarrow  \quad \exists \vex \in \Z^n :\bigvee_{i \in [d]} A_i \vex \leq \veb_i \enspace .
  \]
\end{lemma}
\begin{proof}
  Let $M = B^2 n$, let $y_i$ for $i \in [d]$ be binary variables, and consider the following system:
  \begin{equation*}
    \sum_{i=1}^d y_i = 1 \bigwedge
                 y_i \geq 0, A_i \vex \leq \veb_i + M(1-y_i)~\mbox{for all}~i \in [d] \enspace .
  \end{equation*}
  Assume it has an integer solution $\mathbf{y}$.
  Then there is an index $i \in [d]$ such that $y_i=1$ and thus $A_i \vex \leq \veb_i + \mathbf{0}$ holds; thus, the system $A_i \vex \leq \veb_i$ has an integer solution.
  In the other direction, assume that the system $A_i \vex \leq \veb_i$ has a~solution $\vex$; then let $y_i = 1$.
  We shall prove that $A_{i'} \vex \leq \veb_{i'} + M$ holds for all $i' \neq i$.
  Since $A_{i'} \vex \leq \veb_{i'}$ is $B$\hy small, each of its row sums has $n$ terms which are a~multiple of two numbers, each bounded by $B$, and thus is at most $B^2 n$.
  Moreover, since $y_i=1$, we have $y_{i'} = 0$ and thus the right hand side is $\veb_{i'} + M$ and every assignment of $\vex$ feasible for $A_i \vex \leq \veb_i$ satisfies it.
  Clearly, the new system has $n+d$ variables, $md + d+1$ inequalities, is $B$-bounded and $\|A\|_\infty = B^2 n$ and thus it is $(B^2 n)$\hy small, and can be constructed in the claimed time.
\end{proof}

We now prove Theorem~\ref{thm:main_thm}; we restate it in formal terms here:
\begin{theorem}[formal version of~Theorem~\ref{thm:main_thm}]
\label{thm:main_thm2}
  Let $\Psi$ be a $\beta$-small $\P_{2,\ven,\delta,\gamma,\alpha, \beta}$ sentence (i.e., without free variables and thus $n_0 = 0$).
  Then $\Psi$ can be decided in time $g(\ven, \delta, \gamma) \poly(\alpha, \beta)$ for some computable function $g$.
\end{theorem}
\begin{proof}
  Let $\Psi \equiv \exists \vex_1 \forall \vex_2: \Phi(\vex_1, \vex_2)$.
  Clearly, to decide $\Psi$ we can instead decide $\neg \Psi \equiv \forall \vex_1 \exists \vex_2: \neg \Phi(\vex_1, \vex_2)$.
  Consider the formula $\neg \Phi(\vex_1, \vex_2)$: by Lemma~\ref{lem:woods_dnf}, there exists an equivalent $\beta$-small DNF formula $\bar{\Phi}(\vex_1, \vex_2)$ such that the number of its disjunctions and conjunctions is a function of just the original $\delta$ and $\gamma$.

  Thus, from now on focus on the case $\Psi \equiv \forall \vex_1 \exists \vex_2: \Phi(\vex_1, \vex_2)$.
  Our next task is to construct an instance of \textsc{Parametric ILP} equivalent to deciding $\Psi$.
  To do that, replace $\Phi(\vex_1, \vex_2)$ with $\Xi(\vex_1, \vex_2) \equiv \exists \vex_3: A(\vex_1, \vex_2, \vex_3) \leq \veb$, where $\vex_3$ is of dimension $\delta$ and the system of linear inequalities $A(\vex_1, \vex_2, \vex_3) \leq \veb$ has bounded length, coefficients and right sides.
  Now we use Lemma~\ref{lem:disjunctions}.
  Assume that $\Phi(\vex_1, \vex_2)$ is a disjunction of $\delta$ linear systems, each of at most $\gamma$ conjunctions, and by the assumptions of the theorem we know that $\Psi$ is $\beta$-small.
  Plugging into Lemma~\ref{lem:disjunctions}, we have $B = \beta$, $d = \delta$, $m= \gamma$ and $n=n_1 + n_2$, and we obtain a formula $\exists \vex_3: A(\vex_1, \vex_2, \vex_3) \leq \veb$ equivalent to $\Phi(\vex_1, \vex_2)$.
  Thus, we are left with deciding $\forall \vex_1 \exists (\vex_2, \vex_3): A(\vex_1, \vex_2, \vex_3) \leq \veb$ with the following parameters:
  \begin{itemize}[leftmargin=*]
    \item The coefficients and right-hand sides $\|(A, \veb)\|_\infty$ are bounded by some computable function $f(\ven, \delta, \gamma) \cdot (\alpha \beta)^2$.
    \item The dimensions of $A$ are bounded by $f(\ven, \delta, \gamma)$.
  \end{itemize}
  Thus, we are in the setting of Corollary~\ref{cor:es_pilp} and we can decide the above sentence by a fixed-parameter algorithm.
\end{proof}

\begin{corollary}
\label{cor:main_cor}
  Let $\Psi(\vex_0)$ be a disjunction of $D$ many $\beta$-small $\exists \P_{2,(\tau, n_1, n_2),\delta,\gamma,\alpha, \beta}$ formulas.
  Then {\sc Minimum Move} with objective $f(\vex_0) = \vecc^{\transpose} \vex_0$ and winning condition $\Psi$ can be solved in time $g(\tau, n_1, n_2, \delta, \gamma) \poly(D,\alpha, \beta, \vecc)$ for some computable function $g$.
\end{corollary}
\begin{proof}
  Let $\Psi \equiv \Psi_1 \vee \cdots \vee \Psi_D$, and let $\Xi_i(\vem) \equiv \Psi_i(\ves+\veDelta(\vem)) \wedge (\vecc, B)\textrm{-}\texttt{move}(\vem) \wedge \texttt{feasible}(\veDelta(\vem), \ves)$, for each $i \in [D]$.
  Now we need to perform binary search with parameter $B$ on the sentence
  \[
    \exists \vem \Psi(\ves+\veDelta(\vem) \wedge (\vecc, B)\textrm{-}\texttt{move}(\vem) \wedge \texttt{feasible}(\veDelta(\vem), \ves) \enspace .
  \]
  This sentence is obviously equivalent to
  \[
    \exists \vem \Xi_1(\vem) \vee \cdots \vee \exists \vem \Xi_D(\vem), 
  \]
  which can be decided by $D$ application of Theorem~\ref{thm:main_thm2}.
  This is possible in the claimed time, as, for each $i \in [D]$, $\exists \vem \Xi_i(\vem)$ has quantifier depth $2$, is $\beta'$-small for $\beta' = \poly(\beta, \tau + n_1 + n_2, \|\vecc\|_\infty)$ and has dimension $\ven = (0, \tau + n_1, n_2)$.
  A minimum move can then be constructed by coordinate-wise binary search; cf.~\cite[Thm. 4.1]{KoppeEtAl2010}.
\end{proof}

\subsection{Application: Swap Bribery for the Dodgson Rule and Young Rule}
\label{sec:swapbriberyforthedodgsonruleandyoungrule}
We now prove Theorem~\ref{thm:main-fpt-swapbribery} by giving our fixed-parameter algorithm for $\RR$-{\sc Swap Bribery} with~$\RR$ the Dodgson Rule or the Young Rule.

\begin{proof}[Proof of Theorem~\ref{thm:main-fpt-swapbribery}]
  Fix an instance $(C,V,c^{\star}, \sigma)$ of $\RR$-\textsc{Swap Bribery}.
  As there are $|C|!$ possible total orders on $C$, each voter has one of these orders.
  Thus we view the election as a society $\ves = (s_1, \dots, s_{|C|!})$.
  Let $\vecc_{\textnormal{swap}}$ be a move costs vector defined as $c_{i,j} = $``swap distance between types~$i$ and~$j$'' (this is simply the number of inversions between the permutations~$i$ and~$j$~\cite[Proposition 3.2]{ElkindEtAl2009}); observe that a society $\ves$ is in swap distance at most $d$ from $\ves'$ if $\ves' = \ves + \Delta(\vem)$ and~$\vem$ is a feasible $(\vecc_{\textnormal{swap}}, d)$-move for $\ves$.
  Moreover, let $\vecc_{\textnormal{del}}$ be defined as $c_{i, 0} = 1$ for every type $i$ and $c_{i,j} = +\infty$ for every two types $i,j \neq 0$, where $0$ is a type for latent voters; observe that $\ves$ is in voter deletion distance at most $d$ from $\ves'$ if $\ves' = \ves + \Delta(\vem)$ and $\vem$ is a $(\vecc_{\textnormal{del}}, d)$-\texttt{move}.

  Our plan is to express the winning condition for $\RR$ using a PA formula $\Psi_{\textnormal{Dodgson}}(\ves)$ which is a disjunction of polynomially many (in $|V|$ and $|C|$) formulas from $\exists \P_{2, \ven, \delta, \gamma, \alpha, \beta}$, and then solve {\sc Minimum Move} with $f(\vem) = \bm{\sigma}^{\transpose} \vem$ and $\Psi = \Psi_{\textnormal{Dodgson}}(\ves)$ using Corollary~\ref{cor:main_cor}; recall that since the instance is uniform, we have $\sigma^v \equiv \sigma$ for all $v \in V$, and we let $\bm{\sigma}$ be the move costs vector obtained from $\sigma$.
  For $\Psi_{\textnormal{Young}}$, we proceed analogously.
  In the end, we will verify that $\ven, \delta, \gamma$ are bounded by a function of $|C|$, and that $\alpha, \beta$ are polynomial in the input size.

  \heading{Expressing $\Psi_{\textnormal{Dodgson}}$ and\/ $\Psi_{\textnormal{Young}}$}
  The winning condition for candidate $c^{\star}$ in both the Dodgson and Young rule can be viewed as~$c^{\star}$ being closest to being a~Condorcet winner with respect to some distance measure.
  Specifically, for the Dodgson rule, this distance is the number of swaps in the preference orders, and in the Young rule, the distance is the number of voter deletions.
  For this reason, finding a~bribery which makes $c^{\star}$ a~winner corresponds to finding a~bribery which, for some $d \in \N$, makes $c^{\star}$ be in distance $d$ or less from being a~Condorcet winner, and, simultaneously, making every other candidate $c \neq c^{\star}$ be in distance at least $d$ from being a~Condorcet winner.

  Let us fix $d \in \N$; later we will argue that we can go through all $D$ relevant choices of $d$.
  First, we will express the condition that in a society $\ves$, candidate $c$ beats candidate $c'$ in a head-to-head contest: 
  \[
    \texttt{beats}(c, c', \ves) \equiv \sum_{i: c \pref_i c'} s_i > \sum_{i: c' \pref_i c} s_i,
  \]
  where $\pref_i$ is the preference order shared by all voters of type $i$.
  Then, it is easy to express that $c^{\star}$ is a winner in $\ves$ under Condorcet's rule:
  \[
    \Phi_{\textnormal{Condorcet}}(c, \ves) \equiv \bigwedge_{c' \neq c} \texttt{beats}(c, c', \ves) \enspace .
  \]
  Finally, we express that $c^{\star}$ is a Dodgson-winner with score $d$ by $\Psi_{\textnormal{Dodgson},d}$ as follows:
  \begin{align*}
    \Psi_{\textnormal{Dodgson},d}(\ves) \equiv & \big(\exists \vem: (\vecc_{\textnormal{swap}}, d)\textrm{-}\texttt{move}(\vem) \wedge \texttt{feasible}(\Delta(\vem), \ves) \wedge \\
    & \quad \Phi_{\textnormal{Condorcet}}(c^\star, \ves + \Delta(\vem)) \big) \wedge \\
    & \big(\forall \vem: (\vecc_{\textnormal{swap}}, d-1)\textrm{-}\texttt{move}(\vem) \wedge \texttt{feasible}(\Delta(\vem), \ves) \wedge \\
    & \quad \bigwedge_{c \neq c^{\star}} \neg \Phi_{\textnormal{Condorcet}}(c, \ves + \Delta(\vem)) \big)
  \end{align*}
  Then, $\Psi_{\textnormal{Dodgson}} \equiv \Psi_{\textnormal{Dodgson},1} \vee \cdots \vee \Psi_{\textnormal{Dodgson},D}$.
  Clearly, $\Psi_{\textnormal{Young},d}$ is obtained simply by replacing $\vecc_{\textnormal{swap}}$ with $\vecc_{\textnormal{del}}$.

  \heading{Complexity}
  The number $D$ of relevant choices of $d$ is bounded by $|C|^2 |V|$:
  at most $|C|^2$ swaps suffice for any bribery of a single voter, and there are at most $|V|$ voters, thus $D \leq |C|^2 |V|$ is a bound on the Dodgson score of a candidate.
  For the Young score, $D \leq |V|$.

  Since for each $d \in [D]$, $\Psi_{\textnormal{Dodgson},d} \equiv \exists \vem: \Xi_1(\vem) \wedge \forall \vem: \Xi_2(\vem)$ can be equivalently rewritten as $\exists \vem \forall \vem': \Xi_1(\vem) \wedge \Xi_2(\vem')$, it has quantifier depth $2$, and thus belongs to $\exists \P_{2, \ven, \delta, \gamma, \alpha, \beta}$.
  Let us determine the parameters:
  \begin{itemize}[leftmargin=*]
    \item $\ven$ is the vector of dimensions; thus $n_0 = \tau = |C|!$ are the dimensions of a vector encoding a society, and $n_1 = \tau^2 = |C|!^2$ are the dimensions of a vector encoding a move.
    \item $\delta$ is the number of disjunctions, and it is polynomial in $|C|$,
    \item $\gamma$ is the number of conjunctions, and it is polynomial in $\tau = |C|!$,
    \item $\alpha$ is the largest coefficient, which is $\|\vecc_{\textnormal{swap}}\|_\infty \leq |C|^2$,
    \item $\beta$ is the largest right-hand side, which is $d \leq |C|^2 |V|$.
  \end{itemize}
  Thus, $\ven$, $\delta,\gamma$ are functions of the parameter $|C|$ and $\alpha,\beta$ are polynomial in the size of the input election $(C,V,\{\pref_v\mid v\in V\})$, as required by Corollary~\ref{cor:main_cor}.
  Analogous analysis applies to $\Psi_{\textnormal{Young},d}$.
\end{proof}

Replacing $\vecc_{\textnormal{swap}}$ with $\mathbf{1}^\transpose$ produces the winning condition for Dodgson$'$ as introduced by Faliszewski et al.~\cite{FaliszewskiEtAl2009}.
Furthermore, it is interesting to consider voting rules obtained by replacing Condorcet's rule in the definition of Dodgson's and Young's rule.
For example, the Majority rule also might not produce a winner, and most rules (Copeland$^\alpha$, Scoring protocol etc.) allow ties.
Let the score of a candidate be their distance (swap, deletion, etc.) from being a (unique) winner under rule $\RR$.
We remark that if we replace Condorcet's rule with $\RR$ with a ``simple'' PA desciription, it corresponds to replacing $\Phi_{\textnormal{Condorcet}}$ in the proof above, thus yielding fixed-parameter tractable algorithms for all such rules as well.

\subsection{Application: Resilient and Robust Moves}
\label{subsec:robustresilient}
Theorem~\ref{thm:main_thm} allows us to develop fixed-parameter algorithms for problems related to moves in society.
Problem \textsc{Robust Move} asks for a move that is robust to any adversary move of cost at most~$B_a$:
\prob{\textsc{Robust Move}}
{A society $\ves$, move costs vectors $\vecc$ and $\vecc_a$, a winning condition $\Psi$, budgets $B , B_a \in \N$}
{A $(\vecc, B)$-move $\vem$ such that for every adversary $(\vecc_a, B_a)$-move $\vem_a$, $\Psi(\ves + \Delta(\vem_a) + \Delta(\vem))$ holds.}

The second problem, \textsc{Resilient Budget}, asks if a budget $B$ suffices to counter any adversary move of cost at most~$B_a$.
Crampton et al.~\cite{CramptonEtAl2017} consider a specialization of this problem for $\RR$-\textsc{Swap Bribery}. 

\prob{\textsc{Resilient Budget}}
{A society $\ves$, move costs vectors $\vecc$ and $\vecc_a$, a winning condition $\Psi$, budgets $B , B_a \in \N$}
{Does for every adversary $(\vecc_a, B_a)$-move $\vem_a$ exist a $(\vecc, B)$-move $\vem$ such that $\Psi(\ves + \Delta(\vem_a) + \Delta(\vem))$?}

\begin{theorem}
\label{thm:robustresilient}
  \textsc{Robust Move} and \textsc{Resilient Budget} with $\Psi \in \P_{0, (n_0), \delta, \gamma, \alpha, \beta}$ can be solved in time $g(\tau, n_0, \delta, \gamma) (\alpha + \beta + B + B_a + \|\vecc\|_\infty + \|\vecc_a\|_\infty)^{\Oh(1)}$, that is, \FPT parameterized by $\tau + n_0 + \delta + \gamma$.
\end{theorem}
\begin{proof}
  We apply Theorem~\ref{thm:main_thm} to decide the following formulas, which is clearly equivalent to deciding the problems at hand:
  \begin{align*}
    \Psi(\ves)_{RB} \equiv & \exists \vem: \texttt{feasible}(\ves, \Delta(\vem)) \wedge (\vecc, B)\textrm{-}\texttt{move}(\vem)~ \wedge \\
                           & \forall \vem_a: \texttt{feasible}(\ves + \Delta(\vem), \Delta(\vem_a)) \wedge (\vecc_a, B_a)\textrm{-}\texttt{move}(\vem_a)~ \wedge \\ & \qquad \Psi(\ves + \Delta(\vem) + \Delta(\vem_a)) \\
   \Psi(\ves)_{RM} \equiv  & \forall \vem_a: \texttt{feasible}(\ves, \Delta(\vem_a)) \wedge (\vecc_a, B_a)\textrm{-}\texttt{move}(\vem_a)~ \wedge \\
                           &\exists \vem: \texttt{feasible}(\ves + \Delta(\vem_a), \Delta(\vem)) \wedge (\vecc, B)\textrm{-}\texttt{move}(\vem)~ \wedge \\
                           & \qquad \Psi(\ves + \Delta(\vem) + \Delta(\vem_a)) \qedhere
  \end{align*}
\end{proof}

\section{Polynomially Many Types}
We now prove Theorem~\ref{thm:polymanytypes-informal}, which formally reads:
\begin{theorem}
\label{thm:poly_types}
  Let $\Psi(\ves) \in \P_{0,(\tau, n_1),\delta,\gamma,\alpha, \beta}$  be a winning condition.
  {\sc Minimum Move} can be solved in time $g(\gamma, \alpha) (\tau + \delta)^{\Oh(1)} \log(\ves)$ for any linear function $f(\vem) = \vecc^{\transpose} \vem$, where $g$ is some computable function.
\end{theorem}
\begin{proof}
  Since $k=0$, $\Phi(\vey)$ is a disjunction of $\delta$ linear systems $A_i \vey \leq \veb_i$, with $\|A_i\|_\infty \leq \alpha$ and with at most $\gamma$ rows, for every $i \in [\delta]$.
  Thus we instead solve $\delta$ instances of {\sc Minimum Move} with $\Phi_i(\vey) \equiv A_i \vey \leq \veb_i$ and pick the best solution among them.
  So from now on assume that $\Phi(\vey) \equiv A \vey \leq \veb$ with $\|A\|_\infty \leq \alpha$ and $A \in \Z^{\gamma \times \tau}$.

  Observe that $A$ can have at most $\alpha^{\Oh(\gamma)}$ different columns.
  For two types $i, j \in [\tau]$, we say they are equivalent and write $i \sim j$ if the columns $A_i$ and $A_j$ are identical.
  Thus, the $\tau$ types of people fall into $C \leq \alpha^{\Oh(\gamma)}$ equivalence classes.
  For every type $i \in [\tau]$, let $\CC[i] = \{j \in [\tau] \mid j \sim i\}$ be the equivalence class containing $i$, and let $\bar{A} \in \Z^{\gamma \times C}$ be a matrix with, for every $i \in [C]$, $\bar{A}_{i} = A_j$ where $j \in \CC[i]$.
  Now, for every type $i \in [\tau]$, we shall create a \emph{reduced custom move costs vector} $\vecc^i \in \N^{C^2}$.
  For every $j \in [C]$, $i \neq j$, let 
  \[
    c^i_{i,j} = \min_{j' \in \CC[j]} c_{i,j'}
  \]
  be the cost of moving from $i$ to the cheapest equivalent of $j$, and let $c^i_{j,k} = +\infty$ for any $i \neq j, k \in [C]$.
  Let $\bar{\vecc} = (\vecc^1, \dots, \vecc^\tau)$
 
  Then, consider the following ILP with variables $\vex = (\vex^1, \dots, \vex^\tau) \in \Z^{\tau C}$; we obtain the minimum move $\vem$ from its optimal solution by taking $m_{i,j} = x^i_{j'}$ where $j' \sim j$:
  \begin{align*}
    \min \bar{\vecc}^{\transpose} \vex \quad \mbox{s.t.}
        \sum_{i=1, \ldots, \tau} \bar{A} \vex^i \leq \veb,~
		\sum_{j=1, \ldots, C} x^i_j = s_i   \qquad  \forall i \in [\tau] \enspace .
  \end{align*}
  This ILP is a \emph{combinatorial pre-$n$-fold IP}; by Knop et al.~\cite{KnopEtAl2017b} (detailed in~\cite[Corollary 23]{KnopEtAl2017c}), it can be solved in the claimed time.
\end{proof}

\section{Discussion}
We raise three important questions which naturally arise from this work.
First, we ask whether our fixed-parameter algorithm for $\{\mbox{Dodgson,Young}\}$\hy \textsc{Swap Bribery} with uniform cost extends to general cost functions.
This parallels \cite[Challenge \#2]{BredereckEtAl2014}.
For much simpler voting rules an analogous result was shown only recently~\cite{KnopEtAl2017}.
We believe that if the answer is positive, proving it would require providing new powerful integer programming tools, in particular, some analogue of Theorem~\ref{thm:main_thm} for $n$-fold integer programming, which is the engine behind the recent progress~\cite{KnopEtAl2017}.

Second, can the run time of our algorithm be improved?
This analogously recalls \cite[Challenge \#1]{BredereckEtAl2014}.
The run time of Corollary~\ref{cor:main_cor} is double-exponential in the dimension $\Oh(\tau)$; thus Theorem~\ref{thm:main-fpt-swapbribery} shows that Dodgson-\textsc{Swap Bribery} is solvable in triple-exponential time in parameter $|C|$.
We believe it can be improved, but we are sceptical that it could be made single-exponential, and thus ask if a double-exponential lower bound holds.
A related question is to show lower bounds separating the complexity of $\RR$-\textsc{Swap Bribery} for different voting rules such as Borda, Copeland and Kemeny.

Third, we note that, unlike most previous results, the run time of our algorithm depends polynomially on the number $|V|$ of voters.
Many previous results solve the \emph{succinct} variant of the problem (cf. Falizewski et al.~\cite{FaliszewskiEtAl2009}) and  depend polynomially only on $\log |V|$.
Thus we ask whether $\{\mbox{Dodgson,Young}\}$\hy \textsc{Swap Bribery} is \fpt also in the succinct variant.

Finally, we discuss our usage of Presburger arithmetic as a generalization of ILPs.
Solving ILPs amounts to deciding $\exists \vex: A \vex \leq \veb$; by Lenstra's algorithm~\cite{Lenstra1983} and its improvements by Kannan~\cite{Kannan1987} and Frank and Tardos~\cite{FrankTardos1987}, this task is fixed-parameter tractable parameterized by the dimension of $\vex$ even for unbounded $\|A, \veb\|_\infty$ and if~$A$ has polynomially (in the length of the input) many rows.
In 1990, Kannan claimed to show that \textsc{Parametric ILP}, which amounts to deciding $\forall \veb \in Q~\exists \vex: A \vex \leq \veb$ for some polyhedron~$Q$, is fixed-parameter tractable parameterized by the dimension of $\vex$; here, $\|A\|_\infty$ must be bounded by a polynomial and the number of rows of $A$ also has to be a parameter.
However, Kannan's result relies on \emph{Kannan's Partitioning Theorem} (KPT), which was recently \emph{disproved} by Nguyen and Pak~\cite{NguyenPak2017c}.
Nguyen and Pak~\cite[Theorem 1.9]{NguyenPak2017c} state that Woods~\cite{Woods2015} gave a polynomial-time algorithm for deciding $\forall \vey \exists \vex: \Phi(\vex, \vey)$ when the dimensions of $\vex$ and $\vey$ are \emph{constant}; however, it is unclear if this is a fixed-parameter algorithm.
For this reason, as we have the dimensions of $\vex$ and $\vey$ as (non-constant) \emph{parameter}, we chose to prove Theorem~\ref{thm:main_thm} as a slightly weaker result (still sufficient for our purposes) but using only elementary techniques.

\bibliography{comsoc}

\end{document}